\pdfoutput=1
\RequirePackage{ifpdf}
\ifpdf % We~are running pdfTeX in pdf mode
\documentclass[pdftex]{sigma}
\else
\documentclass{sigma}
\fi

\numberwithin{equation}{section}

\newtheorem{Theorem}{Theorem}[section]
\newtheorem*{Theorem*}{Theorem}
\newtheorem{Corollary}[Theorem]{Corollary}
\newtheorem{Lemma}[Theorem]{Lemma}
\newtheorem{Proposition}[Theorem]{Proposition}
\newtheorem{Conjecture}[Theorem]{Conjecture}
 { \theoremstyle{definition}
\newtheorem{Definition}[Theorem]{Definition}

\newtheorem{Remark}[Theorem]{Remark} }

\begin{document}
\allowdisplaybreaks

\newcommand{\arXivNumber}{2410.21210}

\renewcommand{\PaperNumber}{022}

\FirstPageHeading

\ShortArticleName{Mode Stability of Hermitian Instantons}

\ArticleName{Mode Stability of Hermitian Instantons}

\Author{Lars ANDERSSON~$^{\rm a}$, Bernardo ARANEDA~$^{\rm b}$ and Mattias DAHL~$^{\rm c}$}

\AuthorNameForHeading{L.~Andersson, B.~Araneda and M.~Dahl}

\Address{$^{\rm a)}$~Beijing Institute of Mathematical Sciences and Applications, Beijing 101408, P.R.~China}
\EmailD{\href{mailto:lars.andersson@bimsa.cn}{lars.andersson@bimsa.cn}}

\Address{$^{\rm b)}$~School of Mathematics and Maxwell Institute for Mathematical Sciences, \\
\hphantom{$^{\rm a)}$}~University of Edinburgh, EH9 3FD, UK}
\EmailD{\href{mailto:baraneda@ed.ac.uk}{baraneda@ed.ac.uk}}

\Address{$^{\rm c)}$~Institutionen f\"or Matematik, Kungliga Tekniska H\"ogskolan, 100 44 Stockholm, Sweden}
\EmailD{\href{mailto:dahl@kth.se}{dahl@kth.se}}

\ArticleDates{Received December 13, 2024, in final form March 24, 2025; Published online April 01, 2025}

\Abstract{In this note, we prove the Riemannian analog of black hole mode stability for Hermitian, non-self-dual gravitational instantons which are either asymptotically locally flat (ALF) and Ricci-flat, or compact and Einstein with positive cosmological constant. We show that the Teukolsky equation on any such manifold is a positive definite operator. We also discuss the compatibility of the results with the existence of negative modes associated to variational instabilities.}

\Keywords{gravitational instantons; stability; spinor methods}

\Classification{83C60; 53C25; 32Q26}

\section{Introduction}

A gravitational instanton is a complete and orientable $4$-dimensional, Ricci-flat and asymptotically flat Riemannian manifold. In this note, we are interested in Hermitian asymptotically locally flat (ALF) instantons,
as well as compact Einstein--Hermitian 4-spaces. Recall that a~Riemannian manifold is Hermitian if there is an integrable almost-complex structure which is compatible with the metric. In both the Ricci-flat and Einstein cases, the Goldberg--Sachs theorem adapted to Riemannian signature implies that the above is equivalent to the curvature being algebraically special, see \cite{MR707181, flaherty1976hermitian}.
For the definition of a Riemannian manifold being ALF, see \cite[Definition 1.1]{Biquard:Gauduchon}.
A classification of Hermitian, toric, ALF gravitational instantons has been given by Biquard and Gauduchon \cite{Biquard:Gauduchon}. The toric assumption was recently removed in~\cite{li2023classification}.
In fact an ALF Hermitian non-K\"ahler instanton belongs to the Euclidean Kerr, Taub-bolt, or Chen--Teo families, or is Taub-NUT with the opposite orientation. We shall ignore the last case, since the Taub-NUT manifold is half-flat, see Remark \ref{remark:halfflat} below.
The Euclidean Kerr and Taub-bolt instantons are both Petrov type D, that is both the self-dual and anti-self dual parts of the Weyl tensor are algebraically special and the manifold is non-K\"ahler with respect to both orientations.
The Chen--Teo instanton \cite{2011PhLB..703..359C}, on the other hand is Hermitian but has algebraically general anti-self dual Weyl tensor \cite{aksteiner2021gravitational}. The classification of Hermitian non-K\"ahler ALF instantons complements the classification of hyperk\"ahler instantons \cite{10.4310/ACTA.2021.v227.n2.a2, chen2019gravitational, chen2021gravitational,Kronheimer:1989pu, Minerbe+2011+47+58, 2021arXiv210812991S}, and furthermore bears a close similarity to the classification of compact Einstein--Hermitian non-K\"ahler manifolds by LeBrun \cite{MR2899877}.

Local rigidity of Hermitian ALF instantons was proved by Biquard, Gauduchon, and LeBrun, see~\cite{biquard2023gravitational}. Given a Hermitian gravitational instanton $(M, g_{ab})$, there is an open neighborhood~$\mathcal{O}$ of~$g_{ab}$ in the space of metrics on $M$ such that any gravitational instanton $g'_{ab} \in \mathcal{O}$ is also Hermitian.

\begin{Remark}%\label{remark:integ0}
The above result
leaves open the question of whether a Hermitian ALF gravitational instanton is integrable
in the sense of \cite[Section 12.E]{MR2371700}.
The notion of mode stability for instantons that we study in this paper is a step towards addressing this problem, see below.
\end{Remark}

Let $g_{ab}$ be a Riemannian metric on a 4-dimensional manifold $M$ with Ricci tensor $R_{ab}(g)$ and scalar curvature $S(g)$. We define the Einstein tensor $E_{ab}(g)$ by
\begin{align}\label{eq:EinsteinOp}
E_{ab}(g) = R_{ab}(g) - \frac{S(g)}{4}g_{ab},
\end{align}
see \cite[Section~12.26]{MR2371700}. The metric $g_{ab}$ is Einstein if $E_{ab}(g)=0$. We shall be interested~in~two classes of Einstein metrics: ALF Ricci-flat manifolds, and compact Einstein manifolds with ${S(g)\neq 0}$. Let $g(s)_{ab}$ be a 1-parameter family of metrics, with
\begin{align} \label{eq:LinMetric}
\frac{{\rm d}}{{\rm d}s} g(s)_{ab} |_{s=0} = h_{ab},
\end{align}
and let $E(s)_{ab}$ be the Einstein operator of $g(s)_{ab}$. Assume that $g(0)_{ab}=g_{ab}$ is Einstein and that~$h_{ab}$ is a linearized Einstein perturbation, that is $h_{ab}$ satisfies the equation
\begin{align} \label{eq:DRic}
\frac{{\rm d}}{{\rm d}s} E(s)_{ab}|_{s=0} = 0.
\end{align}
If $(M,g_{ab})$ is ALF Ricci-flat, then $E(s)_{ab}$ in \eqref{eq:DRic} can be replaced by $R(s)_{ab}$, and we say that~$h_{ab}$ is an ALF vacuum perturbation if for any integer $k \geq 0$, $\nabla^k h_{ab} = O\bigl(r^{-1-k}\bigr)$. See Definition~\ref{def:Onotation} below for notation. The ALF Ricci-flat instanton $(M,g_{ab})$ is said to be integrable if for any ALF vacuum perturbation $h_{ab}$, there is a 1-parameter family $g(s)_{ab}$ of ALF Ricci flat metrics such that ${\rm d}/{\rm d}s \, g(s)_{ab} |_{s=0} = h_{ab}$.
Similarly, if $(M,g_{ab})$ is compact Einstein, then we say that it is integrable if for any Einstein perturbation $h_{ab}$, there is a 1-parameter family $g(s)_{ab}$ of Einstein metrics such that ${\rm d}/{\rm d}s \, g(s)_{ab}|_{s=0} = h_{ab}$.

The above notion of integrability (which follows \cite[Section 12.E]{MR2371700}) intends to describe the space of solutions to the Einstein equations around the given solution $g_{ab}$. This space is not necessarily a manifold, so the required curve of metrics may not exist. Integrability in this sense is also known as {\em linearization stability} as defined by Fischer and Marsden \cite{fischer_marsden_1975}. For Lorentzian metrics satisfying the Einstein equations, this concept was introduced by Choquet-Bruhat and Deser~\cite{CHOQUETBRUHAT1973165}, and then thoroughly studied by Fischer, Marsden, Moncrief, and Arms~\mbox{\cite{ARMS198281,Fischer1980, fischer_marsden_1973, Moncrief1, Moncrief2}}. In~particular, for vacuum spacetimes with compact Cauchy hypersurfaces, Moncrief showed~\mbox{\cite{Moncrief1, Moncrief2}} that linearization stability {\em fails} if the solution has Killing vector fields.

A step towards addressing integrability of Hermitian instantons is the problem of mode stability, which is a concept that originates in the study of dynamical stability of Lorentzian black holes, but can be adapted to Riemannian metrics.
Let $(M, g_{ab})$ be Hermitian,
with complex structure $J^{a}{}_{b}$ and unprimed Weyl spinor $\Psi_{ABCD}$ (see Section \ref{sec:preliminaries} for notation).
Then~$J^{a}{}_{b}$ can be represented by a spinor $o^A$ as in \eqref{eq:J} and \eqref{iota} below, with respect to which $\Psi_0=\Psi_{ABCD}o^{A}o^{B}o^{C}o^{D}=0$. This follows from \eqref{SFR} below and its integrability condition, see, for example, the discussion around \cite[equation (12.3)]{Huggett_Tod_1994}.
Let $g(s)_{ab}$ be a 1-parameter family of
metrics on $M$, with $g(0)_{ab} = g_{ab}$.
The linearized $\Psi_0$ is given by
\begin{align} \label{eq:LinPsi0}
\vartheta \Psi_0 = \frac{{\rm d}}{{\rm d}s} \Psi_0[g(s)_{ab}] |_{s=0},
\end{align}
where $\vartheta$ denotes the variation operator introduced in \cite{backdahl2016formalism} (adapted to Riemannian signature) which we shall use in this paper to treat perturbations, see Section \ref{sec:perturbations}. Since the background~$\Psi_0$ vanishes, $\vartheta \Psi_0$ is a gauge invariant quantity.
If the Hermitian manifold $(M,g_{ab})$ is Ricci-flat or Einstein, then for a perturbation satisfying \eqref{eq:DRic}, we have that $\vartheta \Psi_0$ satisfies the Teukolsky equation, see \eqref{TeukEq} below.
The Teukolsky equation in the Riemannian case is the analog of the Teukolsky equation which governs perturbations of the Kerr black hole and other Lorentzian Petrov type D Einstein metrics.
In the Lorentzian case, mode stability means that there are no solutions of the Teukolsky equation with frequency in the upper half plane and satisfying a~condition of no incoming radiation, see \cite{Andersson2016, Whiting}.
In the Riemannian case, we have the following analog of the notion of mode stability.

\begin{Definition} \label{Def:modestability}
Let $(M, g_{ab})$ be Hermitian.
If $(M,g_{ab})$ is ALF Ricci-flat, then we say that mode stability holds for $(M, g_{ab})$ if for any ALF vacuum perturbation we have $\vartheta \Psi_0 = 0$.
If $(M,g_{ab})$ is Einstein and compact, then we say that mode stability holds for $(M, g_{ab})$ if for any Einstein perturbation we have $\vartheta \Psi_0 = 0$.
\end{Definition}

\begin{Remark}
In Lorentzian signature, mode stability for non-vacuum spacetimes such as the Kerr--Newman solution to the Einstein--Maxwell system, or the Kerr--de Sitter black hole in the Einstein case, remains open. In Riemannian signature, Einstein--Maxwell instantons are also of interest, and the problem of their mode stability is also worth studying. The instanton version of the Kerr--de Sitter black hole is a compact Einstein space found by Page \cite{Page}, whose mode stability is proven in Theorem \ref{thm:modestabCompact} below.
\end{Remark}

It was recently shown by Nilsson that mode stability holds for the Petrov type D Euclidean Kerr and Taub-bolt families of instantons, see \cite{2024CQGra..41h5004N}. In this paper, we give a new proof of this result which is valid for all Hermitian ALF instantons, that is including the Chen--Teo case.
Our argument also applies for compact Einstein--Hermitian 4-manifolds.

\begin{Theorem}
\label{thm:modestab}
Let $(M,g_{ab})$ be a Hermitian non-K\"ahler ALF instanton. Then mode stability holds for $(M,g_{ab})$.
\end{Theorem}

\begin{Theorem} \label{thm:modestabCompact}
Let $(M,g_{ab})$ be a compact Einstein--Hermitian $4$-manifold with positive cosmological constant. Then mode stability holds for $(M,g_{ab})$.
\end{Theorem}

\begin{Remark}\label{remark:halfflat}
Recall that a hyperk\"ahler instanton is half-flat. Due to the Hitchin--Thorpe inequality for closed four-manifolds $M$, a Ricci-flat manifold $(M,g_{ab})$ is half-flat if and only if $\tau(M)=\frac{2}{3}\chi(M)$, see, for example, \cite[Section~13.8]{MR2371700}. Since $\tau(M)$ and $\chi(M)$ are topological invariants, any Ricci-flat metric on a closed manifold admitting a half-flat metric will satisfy the same equality and will thus be half-flat.

For ALF manifolds $(M,g_{ab})$, the Hitchin--Thorpe inequality has an additional term giving the contribution from the ALF end, see \cite{Chen-Li21,Dai-Wei07}. Again, equality holds if and only if~$(M,g_{ab})$ is hyperk\"ahler. Since neither $\tau(M)$, $\chi(M)$ nor the contribution at infinity depends on the particular ALF metric, any other Ricci-flat metric will again satisfy the same equality and will thus be half-flat.

This gives the analog of mode stability for hyperk\"ahler instantons. Therefore, we shall consider mode stability only for Hermitian, non-K\"ahler manifolds.
\end{Remark}

Mode stability is a step towards proving the following conjecture.
\begin{Conjecture} \label{conj:integ}
Let $(M, g_{ab})$ be a Hermitian non-K\"ahler ALF instanton. Then $(M,g_{ab})$ is integrable.
\end{Conjecture}

\begin{Remark}%\label{remark:conjecture}
It follows from the classification of Hermitian non-K\"ahler ALF instantons that the corresponding moduli spaces are smooth and 2-dimensional. Therefore, Conjecture~\ref{conj:integ}~can be restated in terms of infinitesimal rigidity, that is the statement that an ALF vacuum perturbation is, modulo gauge, a perturbation with respect to the moduli parameters, see \cite{2024CQGra..41h5004N}.
The conjecture can be addressed by analyzing the compatibility complex along the lines in \cite{AABKW}.
In the Lorentzian case, the corresponding result for Kerr has been shown to hold in \cite{2022arXiv220712952A}.
\end{Remark}

\section{Preliminaries and notation} \label{sec:preliminaries}

In this section, $(M,g_{ab})$ denotes a four-dimensional orientable Riemannian manifold with Levi-Civita connection $\nabla_{a}$.

\subsection{Spinors and complex structures} \label{sec:spinors}
We shall use the formalism of 2-spinors as developed by Penrose and Rindler \cite{PR:I, PR:II}, with the exception that we adapt the framework to Riemannian signature following Woodhouse \cite{Woodhouse_1985} (see also \cite{1994GReGr..26..979G}).
This can be done by first using that the constructions in \cite{PR:I, PR:II} formally apply in a complex space as explained in \cite[Section 6.9]{PR:II}, and then noticing that one can specialize to Riemannian signature by equipping the spin spaces with the Riemannian spinor conjugation~$\dagger$ defined in \cite[equation~(2.5)]{Woodhouse_1985}.

We shall in fact only use spinors up to scale, so the existence of a global spin structure is not assumed. Note that the Taub-bolt and Chen--Teo instantons do not admit spin structures~\mbox{\cite{2011PhLB..703..359C, 1978PhLB...78..249P}}.

Abstract spinor indices are denoted by $A,B,\dots $ and $A',B',\dots $, and can be raised and lowered with the symplectic forms $\epsilon_{AB}$, $\epsilon_{A'B'}$ and their inverses. Tensor indices correspond to pairs of primed and unprimed spinor indices, $a=AA'$, $b=BB'$, etc. For example, the metric is~${g_{ab}=\epsilon_{AB}\epsilon_{A'B'}}$, and the Weyl tensor is
\begin{align*}%\label{eq:Weyltensor}
 W_{abcd} = \Psi_{ABCD}\epsilon_{A'B'}\epsilon_{C'D'} + \tilde\Psi_{A'B'C'D'}\epsilon_{AB}\epsilon_{CD},
\end{align*}
where $\tilde\Psi_{A'B'C'D'}$ and $\Psi_{ABCD}$ are the (totally symmetric) self-dual and anti-self-dual Weyl curvature spinors, respectively.

We shall assume that $(M,g_{ab})$ has an almost-complex structure $J^{a}{}_{b}$ compatible with $g_{ab}$. From in \cite[Chapter~IV, Proposition~9.8]{lawson:michelsohn},
$J^{a}{}_{b}$ can be represented by a projective spinor, say~$o^{A}$. The explicit representation is
\begin{align}\label{eq:J}
 J^{AA'}{}_{BB'} = {\rm i} (o^{A}\iota_{B}+\iota^{A}o_{B})\delta^{A'}_{B'},
\end{align}
see \cite[equation~(4.7)]{Woodhouse_1985}, where $\iota^{A}$ is the complex conjugate of $o^{A}$:
\begin{align}\label{iota}
 \iota^A = o^{\dagger A}.
\end{align}
In \eqref{eq:J}, we chose the normalization $o_{A}\iota^{A}=1$. The pair $\bigl(o^A,\iota^A\bigr)$ is called spin dyad. The~components of $\Psi_{ABCD}$ in this dyad are
\begin{align}\label{weylspinorcomp}
\Psi_0 = \Psi_{ABCD} o^A o^B o^C o^D, \qquad \Psi_1 = \Psi_{ABCD} o^A o^B o^C \iota^D, \qquad \Psi_2 = \Psi_{ABCD} o^A o^B \iota^C \iota^D,
\end{align}
together with \smash{$\Psi_{3}=-\overline{\Psi_1}$} and \smash{$\Psi_{4}=\overline{\Psi_0}$}. Note that $\Psi_2$ is real.

\subsection{Conformally invariant GHP connections} \label{sec:CGHP}

An almost-complex structure is invariant under two kinds of transformations: a rescaling of the spinors in \eqref{eq:J} of the form
\begin{align}\label{GHPscaling}
o^A \to{}& \lambda o^A, \qquad \iota^A \to \lambda^{-1} \iota^A,
\end{align}
where $\lambda\colon M\to U(1)$, and conformal transformations
\begin{align}\label{conformaltmetric}
g_{ab} \to \hat g_{ab} ={}& \Omega^2 g_{ab},
\end{align}
where $\Omega$ is a positive function.

We shall use a framework that is covariant under both of the above transformations. This is closely related to the conformally invariant Geroch--Held--Penrose (GHP) formalism given in~\mbox{\cite[Section 5.6]{PR:I}}, but with three main differences arising from the following requirements: it should apply to arbitrary spinor/tensor fields (the framework in \cite[Section 5.6]{PR:I} applies only to scalar fields), it should be independent of a choice of primed spin dyad, it should be adapted to Riemann signature. A framework satisfying these requirements can be found in \cite{araneda2021conformal}.

Under conformal transformations \eqref{conformaltmetric}, the spin dyad in \eqref{eq:J} transforms as
\begin{align}\label{conformaloiota}
\hat{o}^{A}=\Omega^{-1/2}o^{A}, \qquad \hat\iota^A=\Omega^{-1/2}\iota^A.
\end{align}
A metric-dependent tensor/spinor field $\varphi^{\mathcal{A}}$ which transforms with respect to \eqref{GHPscaling} and \eqref{conformaltmetric} according to
\begin{align} \label{eq:proper}
\varphi^{\mathcal{A}} \to{}& \lambda^p \Omega^w \varphi^{\mathcal{A}}
\end{align}
is said to have conformal weight $w$ and GHP weight $p$. Here, $\mathcal{A}$ represents an arbitrary collection of tensor/spinor indices. We shall refer to fields satisfying \eqref{eq:proper} as properly weighted fields with weights $(w,p)$.
For example, $o^{A}$ and $\iota^{A}$ have weights $\bigl(-\frac{1}{2},1\bigr)$ and $\bigl(-\frac{1}{2},-1\bigr)$, respectively, while~$o_{A}$ and~$\iota_{A}$ have weights $\bigl(\frac{1}{2},1\bigr)$ and $\bigl(\frac{1}{2},-1\bigr)$. The components $\Psi_0$, $\Psi_1$, $\Psi_2$ of the Weyl~spinor (equation~\eqref{weylspinorcomp}) have conformal weight $w=-2$ and GHP weights $4, 2, 0$, respectively.

Let $\chi$ be a scalar field, $\varphi_{A}$ a spinor field, and $v_{a}$ a covector field, all of them with conformal weight $w$ and GHP weight $p$.
We define the covariant derivative $\mathcal{C}_{a}$ by
\begin{align*} %\label{DefC}
&\mathcal{C}_{a}\chi= \nabla_{a}\chi + wf_{a}\chi + p P_{a}\chi, \\
&\mathcal{C}_{AA'}\varphi_{B}= \nabla_{AA'}\varphi_{B} - f_{A'B}\varphi_{A} + wf_{AA'}\varphi_{B} + p P_{AA'}\varphi_{B}, \\
&\mathcal{C}_{a}v_b= \nabla_{a}v_{b} + wf_{a}v_{b} + p P_{a}v_{b} - Q_{ab}{}^{c}v_{c},
\end{align*}
where
\begin{align}
&f_{a} = -\frac{1}{2}J^{c}{}_{b}\nabla_{c}J^{b}{}_{a}, \label{LeeForm} \\
&P_{a} = \omega_{a} -\frac{1}{2} {\rm i} J^{b}{}_{a}f_{b}, \label{1formP} \\
&Q_{ab}{}^{c} = f_{a}\delta^{c}_{b} + f_{b}\delta^{c}_{a} - f^{c}g_{ab}. \nonumber
\end{align}
Here $\omega_{a}=\iota_{B}\nabla_{a}o^{B}$ is the GHP connection 1-form, and $f_{a}$ is the Lee form.
The action of $\mathcal{C}_{a}$ on fields with an arbitrary index structure is defined in the standard way.

\begin{Remark}%\label{Remark:propertiesC}
We have the following facts, which generalize similar results for the standard GHP formalism and its corresponding covariant derivative $\Theta_{a}$.
\begin{enumerate}\itemsep=0pt
\item
If $\chi$ has weights $(w,p)$, then $\bar\chi$ has weights $(w,-p)$. This follows from \eqref{iota}, \eqref{GHPscaling} and~\eqref{conformaloiota}.
This is different from Lorentzian GHP, since in that case there is also a ``$q$-weight'' associated to a primed spin dyad, and complex conjugation interchanges $p$ and $q$.
\item
$\mathcal{C}_{a}$ is real (it commutes with complex conjugation) and metric ($\mathcal{C}_{a}g_{bc}=0$).
Reality can be seen by making use of the previous item together with the fact that the 1-form $P_{a}$ in \eqref{1formP} is purely imaginary, $\bar{P}_{a}=-P_{a}$.
\item
$\mathcal{C}_{a}$ is covariant under GHP and conformal transformations \cite[Section 2.3]{araneda2021conformal}: if \eqref{eq:proper} holds, then
\begin{align*}
 \mathcal{C}_{a}\varphi^{\mathcal{A}} \to \lambda^p \Omega^w \mathcal{C}_{a}\varphi^{\mathcal{A}}.
\end{align*}
This generalizes the transformation rule of the standard GHP derivative, that is $\Theta_{a}\varphi^{\mathcal{A}} \to \lambda^p\Theta_{a}\varphi^{\mathcal{A}}$ under GHP scalings \eqref{GHPscaling}.
\item
$(M,g_{ab})$ is Hermitian if and only if $\mathcal{C}_{a}o^{B}=0$. This follows since
\begin{align*}
\mathcal{C}_{AA'}o^{B} = \sigma_{A'}\iota_{A}\iota^{B}, \qquad
\sigma_{A'}=o^{A}o^{B}\nabla_{AA'}o_{B},
\end{align*}
and, using GHP notation \cite{PR:I},
\begin{align}\label{SFR}
o^{A}o^{B}\nabla_{AA'}o_{B} = 0 \qquad \Leftrightarrow \qquad \kappa=\sigma=0,
\end{align}
which is the condition for the existence of a shear-free null geodesic congruence \cite[equation~(7.3.1)]{PR:II}, or, equivalently, an integrable complex structure \cite[Section 2.4]{araneda2021conformal}.
This generalizes the characterization of K\"ahler manifolds in terms of the GHP connection: a Riemannian 4-manifold is K\"ahler if and only if $\Theta_{a}o^{B}=0$, see \cite[Chapter~IV, Proposition~9.8]{lawson:michelsohn}.
\item
$(M,g_{ab})$ is conformally K\"ahler iff it is Hermitian and $f_{a}=\nabla_{a}\log\phi$ for some scalar field $\phi$ (with $w=-1$, $p=0$). This field satisfies
\begin{align*}
\mathcal{C}_{a}\phi=0.
\end{align*}
For example, in the Einstein--Hermitian case we have that
\cite[Remark 5.1]{araneda2021conformal}
\begin{align} \label{phiEinstein}
\phi \propto \Psi^{1/3}_{2}.
\end{align}
\item
If $(M,g_{ab})$ is conformally K\"ahler, and $u_{a}$ has weights $w$, $p=0$, then
\begin{align} \label{identity}
\mathcal{C}_{a}u^{a} = \phi^{-(w+2)}\nabla_{a}\bigl(\phi^{w+2}u^{a}\bigr).
\end{align}
\end{enumerate}
\end{Remark}

\subsection{Perturbations and the Teukolsky equation}
\label{sec:perturbations}

Here we introduce some notation for gravitational perturbations and prove an identity that will be needed in Section \ref{sec:modestability}.

The gravitational perturbations we consider are of two types: either compactly supported in the compact Einstein case, or they satisfy certain fall-off conditions in the ALF Ricci-flat case. For the latter, we use the following notation, which is taken from \cite{2023arXiv230614567A}.
\begin{Definition}[{\cite[Definition~2.6]{2023arXiv230614567A}}]\label{def:Onotation}
Let $(M,g_{ab})$ be an ALF manifold as defined in \cite[Definition~2.1]{2023arXiv230614567A}. Let $t$ and $s$ be any two tensor fields on $(M,g_{ab})$. We write
\begin{align*}%\label{Onotation}
 t = O(r^{\alpha}), \qquad s = O^{*}(r^{\alpha})
\end{align*}
if there is a constant $C$ such that $|t|\leq C r^{\alpha}$ for $r\geq A$, and $\bigl|\nabla^{k}s\bigr|=O\bigl(r^{\alpha-k}\bigr)$ for all non-negative integers $k$. Here, \smash{$|t|^2=t_{a\dots d}\overline{t}{}^{a\dots d}$}.
\end{Definition}

We treat variations of spinor and tensor fields following the approach introduced in \cite{backdahl2016formalism},
which can be adapted to Riemannian signature.
In particular, given a symmetric 2-tensor $h_{ab}$ on $(M, g_{ab})$, viewed as a linear perturbation of the metric, the corresponding perturbations of spinor and tensor fields are given by the variation operator $\vartheta$. For example, the variation of the unprimed Weyl spinor is $\vartheta \Psi_{ABCD}$, and the variations of the scalars $\Psi_i$, $i=0,1,2$ are given by~$\vartheta \Psi_i$.
The formula for $\vartheta \Psi_{ABCD}$ is
\begin{align}\label{linWeylSpinor}
\vartheta\Psi_{ABCD}
= \frac{1}{2}\nabla_{(A}^{A'}\nabla_{B}^{B'}h_{CD)A'B'} + \frac{1}{4}g^{ef}h_{ef}\Psi_{ABCD},
\end{align}
which coincides with \cite[equation~(5.7.15)]{PR:I}.

Let $h_{ab}$ be an arbitrary metric perturbation. Recall that we defined the Einstein operator in~\eqref{eq:EinsteinOp}. We denote its linearization by $\vartheta{E}_{ab}$, and the linearized Ricci tensor and scalar curvature by $\vartheta R_{ab}$ and $\vartheta S$. From \cite[Theorem~1.174]{MR2371700}, we have
\begin{align}
\vartheta E_{ab} &{}= \vartheta R_{ab} - \frac{1}{4}\vartheta S g_{ab} - \frac{1}{4}S(g)h_{ab} \nonumber\\
&{}= \frac{1}{2}\Delta h_{ab} - \frac{1}{2}\nabla_{a}\nabla_{b}\bigl(g^{cd}h_{cd}\bigr) + \frac{1}{2}\nabla^{c}\nabla_{a}h_{bc} + \frac{1}{2}\nabla^{c} \nabla_{b}h_{ac} \nonumber\\
&\quad{} - \frac{1}{4}g_{ab}\bigl[ \nabla^{c}\nabla^{d}h_{cd} + \Delta\bigl(g^{cd}h_{cd}\bigr) - h^{cd}R_{cd} \bigr]
-\frac{1}{4}S(g) h_{ab},\label{LinEinsteinOp}
\end{align}
where $\Delta=-g^{ab}\nabla_{a}\nabla_{b}$. Define the operator
\begin{align}\label{TeukOp}
 {\rm L}=g^{ab}\mathcal{C}_{a}\mathcal{C}_{b}-18\Psi_{2}
\end{align}
acting on scalar fields of weight $(w,p)$. We have the following lemma.
\begin{Lemma}\label{lem:OpId}
Let $(M,g_{ab})$ be Einstein--Hermitian, with $($possibly vanishing$)$ cosmological constant $\lambda$. Let $h_{ab}$ be an arbitrary metric perturbation \eqref{eq:LinMetric}, $\vartheta{E}_{ab}$ the linearized Einstein operator~\eqref{LinEinsteinOp}, and $\vartheta\Psi_0$ the linearized Weyl scalar~\eqref{eq:LinPsi0}. Furthermore, let $f_{a}$ be the Lee form \eqref{LeeForm}, and $Q^{abcd}$ the tensor field
\begin{align}\label{tensorQ}
Q^{abcd} = o^{A}o^{B}o^{C}o^{D}\epsilon^{A'B'}\epsilon^{C'D'}.
\end{align}
Then
\begin{align}\label{eq:OpIdentity}
-\mathring{\Omega}^{-1}Q^{acbd}(\nabla_{a}-4f_{a})\nabla_{d}\vartheta E_{bc} = {\rm L}\bigl[\mathring{\Omega}^{-1}\vartheta\Psi_0\bigr],
\end{align}
where ${\rm L}$ is the operator \eqref{TeukOp} and $\mathring{\Omega}$ is an auxiliary constant conformal factor, that is a scalar field with weights $w=1$, $p=0$ and $\nabla_{a}\mathring\Omega=0$.
\end{Lemma}

\begin{Corollary}\label{item:teuk}
Let $h_{ab}$ be a linearized Einstein perturbation $\vartheta E_{ab}=0$, and let $\chi=\mathring{\Omega}^{-1}\vartheta\Psi_0$. Then $\chi$ solves the Teukolsky equation
\begin{align} \label{TeukEq}
{\rm L}[\chi] = 0.
\end{align}
\end{Corollary}

\begin{Remark} \label{remark:OpId}
\quad
\begin{enumerate}\itemsep=0pt
\item
The auxiliary constant conformal factor $\mathring\Omega$ is necessary for conformal invariance \big(note that~\smash{${\nabla_{a}\mathring\Omega=0}$} but \smash{$\mathcal{C}_{a}\mathring\Omega\neq0\big)$}, see the proof below. Once the operator $\mathcal{C}_{a}$ is written in terms of the ordinary Levi-Civita connection (see \eqref{eq:LNP} below), one can set $\mathring{\Omega}= 1$.
\item\label{item2-2.5}
If $(M,g_{ab})$ is Hermitian, and $\chi$ has $w(\chi)=-3$ and $p(\chi)=4$, then in Newman--Penrose notation, we have
\begin{align} \label{eq:LNP}
{\rm L}[\chi]&{}=2\bigl[(D+\tilde\varepsilon-3\varepsilon-4\rho-\tilde\rho)\bigl(D'+4\varepsilon'-\rho'\bigr) \nonumber \\
&\quad{}
 -\bigl(\delta+\tilde\beta'-3\beta-4\tau-\tau'\bigr)\bigl(\delta'+2\beta'-\tau'\bigr) - 3\Psi_2 \bigr]\chi,
\end{align}
so we see that ${\rm L}$ coincides with the Teukolsky operator \cite[equation~(2.12)]{Teukolsky1973}.
\end{enumerate}
\end{Remark}

\begin{proof}[Proof of Lemma \ref{lem:OpId}]
The strategy is to consider identities for an arbitrary Riemannian manifold $(M,g_{ab})$, and then to linearize around an Einstein--Hermitian metric. Consider then an arbitrary $(M,g_{ab})$, with Levi-Civita connection $\nabla_{a}$ and unprimed Weyl curvature spinor $\Psi_{ABCD}$. Let $J^{a}{}_{b}$ be a (locally defined) compatible almost-complex structure, and let $o^{A}$ be the associated spinor field as in Section \ref{sec:spinors}. We can then define conformally and GHP weighted fields as in Section \ref{sec:CGHP}, together with the connection $\mathcal{C}_{a}$ on the corresponding bundles.

Let $\mathring\Omega$ be an arbitrary constant conformal factor, that is a scalar field with weights~${w=1}$, ${p=0}$ and $\nabla_{a}\mathring\Omega=0$, and define
\begin{align}\label{GravSpin2Field}
 \varphi_{ABCD} := \mathring\Omega^{-1}\Psi_{ABCD}.
\end{align}
This object has weights $w=-1$, $p=0$, and it is essentially the ``gravitational spin 2 field'' of Penrose and Rindler \cite[equation~(9.6.40)]{PR:II}. We have
\begin{align*}
 \mathcal{C}_{AA'}\mathcal{C}^{A'E}\varphi_{BCDE} &=
(\nabla_{AA'}-4f_{AA'})\nabla^{A'E}\varphi_{BCDE} \\
 &= \mathring\Omega^{-1}(\nabla_{AA'}-4f_{AA'})\nabla^{A'E}\Psi_{BCDE} \\
 &= \mathring\Omega^{-1}(\nabla_{AA'}-4f_{AA'})\nabla_{(B}^{B'}\Phi_{CD)B'}{}^{A'},
\end{align*}
where in the second line we used the definition \eqref{GravSpin2Field} and the fact that $\mathring\Omega$ is constant, and in the third line we used the spinor form of the Bianchi identities \cite[equation~(4.10.7)]{PR:I}
adapted to Riemann signature (recall Section \ref{sec:spinors}).
Here $\Phi_{ABA'B'}$ is the trace-free Ricci spinor. Contracting with $o^Ao^Bo^Co^D$:
\begin{align}\label{AuxId}
\nonumber o^Ao^Bo^Co^D\mathcal{C}_{AA'}\mathcal{C}^{A'E}\varphi_{BCDE} &{} = \mathring\Omega^{-1}o^Ao^Bo^Co^D (\nabla_{AA'}-4f_{AA'})\nabla_{B}^{B'}\Phi_{CDB'}{}^{A'} \\
\nonumber &{} = \mathring\Omega^{-1}o^Ao^Co^Bo^D\epsilon^{A'C'}\epsilon^{B'D'} (\nabla_{AA'}-4f_{AA'})\nabla_{DD'}\Phi_{BCB'C'}\\
&{} = -\frac{1}{2}\mathring\Omega^{-1}Q^{acbd}(\nabla_{a}-4f_{a})\nabla_{d}E_{bc},
\end{align}
where in the last line we used the definition \eqref{tensorQ} and the identity \smash{$\Phi_{bc}=-\frac{1}{2}E_{bc}$} (see \cite[equa\-tion~(4.6.25)]{PR:I} and recall \eqref{eq:EinsteinOp}).

It remains to find a convenient expression for the first line in \eqref{AuxId}. This can be done using \cite[equations~(3.6) and (3.15)]{araneda2018conformal}. We have
\begin{align}\label{AuxId2}
o^Ao^Bo^Co^D\mathcal{C}_{AA'}\mathcal{C}^{A'E}\varphi_{BCDE} =
\frac{1}{2} \bigl(g^{ab}\mathcal{C}_{a}\mathcal{C}_{b}-18\Psi_{2}\bigr)\varphi_{0} + B,
\end{align}
where $\varphi_{0}= o^Ao^Bo^Co^D\varphi_{ABCD}=\mathring{\Omega}^{-1}\Psi_0$ and $B$ is a term which couples the GHP quantities $\kappa$, $\sigma$, $\Psi_0$, $\Psi_1$ quadratically. Equating the last line of~\eqref{AuxId} to the right-hand side of~\eqref{AuxId2}, and taking a~linearization around a metric which satisfies $E_{ab}|_{s=0}=0$, ${\Psi_{0}|_{s=0}=\Psi_{1}|_{s=0}=\kappa|_{s=0}=\sigma|_{s=0}=0}$, that is, an Einstein--Hermitian metric, see \eqref{SFR}, the result \eqref{eq:OpIdentity} follows.
\end{proof}

\begin{Lemma}\label{lem:divergence}
Let $(M, g_{ab})$ be Einstein--Hermitian, with volume form ${\rm d}\mu$.
Let $V$ be a four-dimen\-sion\-al region in $M$ with boundary $\partial{V}$, whose unit normal and induced volume form are~$n^{a}$,~${\rm d}\Sigma$, respectively. For any scalar field $\chi$ with conformal weight $w=-3$ and GHP weight~${p=4}$ satisfying~\eqref{TeukEq}, we have
\begin{align} \label{eq:MainIdentity}
0 =
\int_{\partial{V}} \Psi_{2}^{-4/3}\bar\chi(n^a\mathcal{C}_{a}\chi) \, {\rm d}\Sigma
- \int_{V}\Psi_{2}^{-4/3}\bigl(|\mathcal{C}\chi|^{2} + 18\Psi_{2}|\chi|^2\bigr){\rm d}\mu.
\end{align}
\end{Lemma}

\begin{proof}
We have that $(M,g_{ab})$ is Einstein and conformally K\"ahler, see \cite{MR707181}, so identities \eqref{phiEinstein} and \eqref{identity} hold.
Let $\chi$ be a solution to \eqref{TeukEq}, with $w=-3$, $p=4$.
First notice that the covector field $\bar\chi\mathcal{C}_{a}\chi$ has weights $w=-6$, $p=0$, so using \eqref{phiEinstein} and \eqref{identity}, we have
\begin{align} \label{identity2}
\mathcal{C}_{a}(\bar\chi\mathcal{C}^{a}\chi) = \Psi_{2}^{4/3}\nabla_{a}\bigl[\Psi_{2}^{-4/3}\bar\chi\mathcal{C}^{a}\chi\bigr].
\end{align}
Now we multiply \eqref{TeukEq} by \smash{$\Psi_{2}^{-4/3}\bar\chi$} and use the Leibniz property of $\mathcal{C}_{a}$ together with \eqref{identity2},
\begin{align*}
 0 ={}& \Psi_{2}^{-4/3}\bar\chi{\rm L}[\chi]
 =  \Psi_{2}^{-4/3}\bar\chi g^{ab}\mathcal{C}_{a}\mathcal{C}_{b}\chi - 18\Psi_{2}^{-1/3}|\chi|^2 \\
 ={}& \Psi_{2}^{-4/3}g^{ab}\mathcal{C}_{a}(\bar\chi \mathcal{C}_{b}\chi) - \Psi_{2}^{-4/3}g^{ab}(\mathcal{C}_{a}\bar\chi)(\mathcal{C}_{b}\chi) - 18\Psi_{2}^{-1/3}|\chi|^2 \\
 ={}& \nabla_{a}\bigl(\Psi_{2}^{-4/3}\bar\chi \mathcal{C}^{a}\chi\bigr) - \Psi_{2}^{-4/3}\bigl(|\mathcal{C}\chi|^{2} + 18\Psi_{2}|\chi|^2\bigr).
\end{align*}
Integrating this equation over a four-dimensional region $V$ and using the divergence theorem, we get \eqref{eq:MainIdentity}.
\end{proof}

\section{Mode stability}
\label{sec:modestability}

\subsection{ALF instantons}
The proof of the following lemma is similar to the proof of \cite[Theorem~A]{biquard2023gravitational}.
\begin{Lemma}\label{lem:Psi2}
Let $(M, g_{ab})$ be a Hermitian non-K\"ahler ALF instanton. Then $\Psi_2 > 0$ in $M$.
\end{Lemma}
\begin{Remark}\label{rem:Psi2}\quad
\begin{enumerate}\itemsep=0pt
\item
In view of the classification of Hermitian non-K\"ahler ALF instantons \cite{li2023classification}, one could prove Lemma \ref{lem:Psi2} by an explicit calculation for the relevant families of instantons. For the Chen--Teo case, the calculation needed is lengthy but can be done along the lines in \cite{aksteiner2021gravitational}.
\item \label{point:Psi2}
For a conformally K\"ahler manifold $(M,g_{ab})$, where the K\"ahler metric and its scalar curvature are $\hat{g}_{ab}=\varphi^2 g_{ab}$ and $\hat{S}$, it holds
\begin{align}\label{Psi2Scal}
\Psi_{2}=\varphi^2 \frac{\hat{S}}{12},
\end{align}
see \cite{araneda2023}. Thus $\operatorname{sign}\Psi_2=\operatorname{sign}\hat{S}$, so it is sufficient to show that $\hat{S}>0$.
\end{enumerate}
\end{Remark}
\begin{proof}
Let $\mathcal{W}^+$ be the self-dual part of the Weyl tensor.
By \cite[Proposition~5, p.~420]{MR707181}, we have that $\mathcal{W}^+$ does not have zeros in $M$, so $\Psi_2$ does not have zeros either. Hence, by \eqref{Psi2Scal}, $\hat{S}$~does not change sign.
With the conformal factor
\begin{align}\label{CF}
\varphi = 24^{1/6} \bigl|\mathcal{W}^+\bigr|_g^{1/3}
\end{align}
the metric $\hat g_{ab} = \varphi^2 g_{ab}$ is extremal K\"ahler with scalar curvature $\hat{S}$ satisfying
\begin{align*}
\hat{S} \varphi^3 = 6 \Delta \varphi,
\end{align*}
where $\Delta = -g^{ab} \nabla_a \nabla_b$. By construction, $\varphi > 0$. Recalling Definition \ref{def:Onotation}, we have $\bigl|\mathcal{W}^+\bigr| = O\bigl(r^{-3}\bigr) $ so $\varphi \to 0$ at $\infty$. Hence $\varphi$ must have a local maximum at some $x \in M$, and $\Delta \varphi {|}_{x} \geq 0$. This implies $\hat S \varphi^3 \big{|}_x > 0$, which since $\varphi > 0$ implies~${\hat S(x) > 0}$.
By point \eqref{point:Psi2} of Remark~\ref{rem:Psi2}, we find that~${\Psi_2 > 0}$.
\end{proof}

We are now ready to prove our main theorem.

\begin{proof}[Proof of Theorem \ref{thm:modestab}]
Consider an ALF vacuum perturbation $h_{ab}$, that is, $h_{ab}$ satisfies \eqref{eq:DRic} and $\nabla^k h_{ab} = O\bigl(r^{-1-k}\bigr)$ for any integer $k \geq 0$. We recall that the symbols $O$, $O^{*}$ used here and below were introduced in Definition \ref{def:Onotation}.
Let $\chi = \vartheta \Psi_0$ be the linearized extreme Weyl scalar. Since this involves two derivatives of $h_{ab}$ (see \eqref{linWeylSpinor}), the ALF assumption for the perturbation implies $\chi = O^*\bigl(r^{-3}\bigr)$.
In particular, we have $\mathcal{C}_a \chi = O\bigl(r^{-4}\bigr)$. In addition, the ALF condition for the background instanton implies \smash{$\Psi_2 = O\bigl(r^{-3}\bigr)$}, so \smash{$\Psi_2^{-4/3} = O\bigl(r^4\bigr)$}. Therefore,~letting ${V = \{r < R\}}$, from the above we deduce that on $\partial V$ we have
\begin{align*}
\Psi_{2}^{-4/3}\bar\chi(\mathcal{C}_{a}\chi) =O\bigl(r^{-3}\bigr)
\end{align*}
while $A(\partial V) = O\bigl(r^2\bigr)$. This shows that the boundary term in \eqref{eq:MainIdentity} is $O\bigl(r^{-1}\bigr)$ and hence letting $r \to \infty$, we have
\begin{align*}
0 = \int_{M}\Psi_{2}^{-4/3}\bigl(|\mathcal{C}\chi|^{2} + 18\Psi_{2}|\chi|^2\bigr){\rm d}\mu.
\end{align*}
Since by Lemma \ref{lem:Psi2} $\Psi_2 > 0$, we get $\chi=0$ and the result follows.
\end{proof}

\subsection{The compact case}
In this section, we extend our mode stability result (Theorem~\ref{thm:modestab})
to the compact case. The only known Ricci-flat compact 4-manifolds are the flat 4-torus and K3 surfaces (which are half-flat), so we need to include a cosmological constant $\lambda\neq0$ (see Remark~\ref{remark:halfflat}).

A classification of compact Einstein--Hermitian (non-K\"ahler) 4-manifolds with $\lambda>0$ is known from LeBrun \cite{MR2899877}, the only possibilities are the Fubini--Study metric on \smash{$\mathbb{CP}^2$} (with orientation opposite to the K\"ahler one), the Page metric on \smash{$\mathbb{CP}^2\#\overline{\mathbb{CP}}^2$}, or the Chen--LeBrun--Weber metric on \smash{$\mathbb{CP}^2\#2\overline{\mathbb{CP}}^2$}. We note that the Page metric corresponds to a special limit of the Riemannian Kerr--de Sitter solution \cite{Page}.

\begin{proof}[Proof of Theorem \ref{thm:modestabCompact}]
Let $(M,g_{ab})$ be a compact Einstein--Hermitian 4-manifold with $\lambda>0$, and consider a metric perturbation $h_{ab}$. Note that Lemmas \ref{lem:OpId} and \ref{lem:divergence} apply also to the compact case. By item~(\ref{item2-2.5}) in Remark~\ref{remark:OpId}, if $h_{ab}$ solves $\vartheta E_{ab}=0$, then we have a solution $\chi$ to equation~\eqref{TeukEq}, so identity \eqref{eq:MainIdentity} applies. From the above list of compact Einstein--Hermitian instantons, we see that all of them are closed, so the boundary term in \eqref{eq:MainIdentity} vanishes:
\begin{align*}%\label{eq:MainIdentityCompact}
0 = \int_{M}\Psi_{2}^{-4/3}\bigl(|\mathcal{C}\chi|^{2} + 18\Psi_{2}|\chi|^2\bigr){\rm d}\mu.
\end{align*}
The sign of $\Psi_2$ can be determined by an analog of Lemma \ref{lem:Psi2}: from item \eqref{point:Psi2} in Remark \ref{rem:Psi2}, we need only focus on the sign of the scalar curvature $\hat{S}$ of the conformally related K\"ahler metric $\hat{g}_{ab}=\varphi^{2}g_{ab}$, where $\varphi$ is still given by \eqref{CF}. Using \cite[equation~(D.9)]{MR757180} with $\Omega=\varphi$, the conformal behaviour of scalar curvature is
\begin{align*}
 \varphi^2\hat{S} = S + 6\varphi^{-1}\Delta\varphi.
\end{align*}
Since $S=4\lambda>0$, and since the proof of Lemma \ref{lem:Psi2} applies to show that $\varphi^{-1}\Delta\varphi>0$, we have~${\hat{S}>0}$, and thus $\Psi_2>0$. So $\chi=0$, and the result follows.
\end{proof}

\subsection{Negative modes} %\label{Sec:negativemodes}
Here we comment on the compatibility of our mode stability results with other notions of stability in the literature, both in the ALF and compact cases.

A frequently used definition of Riemannian linear stability for Einstein metrics, see, for example, \cite[Definition~4.63]{MR2371700}, is in terms of a variational problem: given the Einstein--Hilbert functional $S$, an Einstein metric $g_{ab}$ is said to be stable if the second variation of $S$ at $g_{ab}$ is negative for all compactly supported, trace-free metric perturbations. If, on the other hand, one can find a perturbation such that the second variation of $S$ is positive, then $g_{ab}$ is said to be unstable.

The above definition is often formulated as an eigenvalue problem: if $h_{ab}$ satisfies the TT conditions $\nabla^{a}h_{ab} = 0$ and $g^{ab}h_{ab} = 0$, one considers the problem $L(h)_{ab}=\mu h_{ab}$, where $L(h)_{ab}=-g^{cd}\nabla_{c}\nabla_{d}h_{ab}-2R_{a}{}^{c}{}_{b}{}^{d}h_{cd}$, and the solution is unstable if there is a negative mode $\mu<0$. See, for example, \cite[Section V]{Gross}, where a negative mode is found for the Schwarzschild instanton, and used to argue about the {\em semi-classical} instability of the solution; see also Witten's work~\cite{Witten}.

It was recently shown in \cite{BiquardOzuch} that if $(M,g_{ab})$ is a conformally K\"ahler 4-manifold which is either compact and Einstein, or ALF and Ricci-flat, then it is unstable in the above variational sense. Here we point out the following.
\begin{Proposition}
The unstable metric perturbations found in {\rm \cite{BiquardOzuch}} are conformally half-flat: the unprimed linearized Weyl curvature spinor identically vanishes.
\end{Proposition}

\begin{Remark}
The above result means that $\vartheta\Psi_{ABCD}=0$, thus in particular $\vartheta\Psi_0=0$, so we see that the variational instability is still compatible with mode stability in the sense of Definition~\ref{Def:modestability}.
\end{Remark}

\begin{proof}
In both the compact and ALF cases, the unstable metric perturbations in \cite{BiquardOzuch} are given by the composition of a closed anti-self-dual 2-form $\omega^{-}$ and the conformal Killing--Yano tensor $\tau$ associated to the conformal K\"ahler structure; see \cite{BiquardOzuch}. In spinor notation, this can be expressed as follows: $\omega^{-}_{ab}=\phi_{A'B'}\epsilon_{AB}$, $\tau_{ab}=K_{AB}\epsilon_{A'B'}$, and the unstable perturbation is
\begin{align} \label{dest}
h_{ab} = (\omega^{-}\circ\tau)_{ab} = \phi_{A'B'}K_{AB},
\end{align}
where $\phi_{A'B'}$ and $K_{AB}$ satisfy the Maxwell and Killing spinor equations, respectively,
\begin{align} \label{MKS}
\nabla^{AA'}\phi_{A'B'} = 0, \qquad \nabla_{A'(A}K_{BC)}=0.
\end{align}
In the compact case, $\omega^{-}_{ab}$ can be any closed anti-self-dual 2-form, that is any Maxwell field $\phi_{A'B'}$. In the ALF case, $\phi_{A'B'}=\nabla_{A(A'}X^{A}_{B')}$, where $X^a$ is the Killing field associated to the Killing spinor $K_{AB}$ \cite{BiquardOzuch}.

We can compute the linearized Weyl spinor using formula \eqref{linWeylSpinor}. Since the trace of \eqref{dest} vanishes, we have
\begin{align*}
\vartheta\Psi_{ABCD}
= \frac{1}{2}\nabla_{(A}^{A'}\nabla_{B}^{B'}\bigl[K_{CD)}\phi_{A'B'}\bigr]
= 0
\end{align*}
where the second equality follows from \eqref{MKS}.
\end{proof}

\subsection*{Acknowledgements}

This work was initiated while the authors were participating in the conference ``Einstein Spaces and Special Geometry'' at Institut Mittag-Leffler in Djursholm, Sweden, in July 2023. BA acknowledges support of the Institut Henri Poincar\'e (UAR 839 CNRS-Sorbonne Universit\'e) and LabEx CARMIN (ANR-10-LABX-59-01) in Paris, during a research stay.

\pdfbookmark[1]{References}{ref}
\LastPageEnding

\end{document}